\documentclass{article}
\usepackage[utf8]{inputenc}

\usepackage{amsmath}
\usepackage{amssymb}
\usepackage{amsthm}
\usepackage{tikz}
\usepackage{graphicx}
\usepackage{float}
\usepackage{algorithm}
\usepackage[noend]{algpseudocode}
\usepackage{tabularx}
\usepackage[margin=1in]{geometry}
\usepackage{arydshln}
\usepackage{multirow}
\usepackage{wrapfig}
\usepackage{xcolor}
\usepackage{xspace}
\usepackage{hyperref}
\usepackage{comment}
\usepackage{bbm}

\usetikzlibrary{positioning,calc,shapes,arrows}
\usetikzlibrary{backgrounds}
\usetikzlibrary{matrix,shadows,arrows} 
\usetikzlibrary{decorations.pathreplacing}

\def\etal.{et\penalty50\ al.}
\theoremstyle{plain}
\newtheorem{theorem}{Theorem}[section]
\newtheorem{lemma}[theorem]{Lemma}

\newtheorem{fact}[theorem]{Fact}

\newtheorem{openprob}{Open Problem}

\theoremstyle{definition}

\theoremstyle{remark}

\theoremstyle{plain}
\newtheorem*{theorem*}{Theorem}

\newcommand*\samethanks[1][\value{footnote}]{\footnotemark[#1]}

\algnewcommand{\LineComment}[1]{\State \(\triangleright\) #1}

\def\etal.{et\penalty50\ al.}

\def\IR {\mathbb{R}}

\title{Maximum Weight Convex Polytope}

\author{Mohammad Ali Abam\\
Sharif University of Technology, CE \\
\href{mailto:abam@sharif.edu}{abam@sharif.edu}\and
Ali Mohammad Lavasani\\
Concordia University, CSSE\\
\href{mailto:ali.mohammadlavasani@concordia.ca}{ali.mohammadlavasani@concordia.ca\thanks{Research is supported by NSERC.}}\and
Denis Pankratov\\
Concordia University, CSSE\\
\href{mailto:denis.pankratov@concordia.ca}{denis.pankratov@concordia.ca}\samethanks}

\date{\today}

\begin{document}

\maketitle

\begin{abstract}
We study the maximum weight convex polytope problem, in which the goal is to find a convex polytope maximizing the total weight of enclosed points. Prior to this work, the only known result for this problem was an $O(n^3)$ algorithm for the case of $2$ dimensions due to Bautista et al. We show that the problem becomes $\mathcal{NP}$-hard to solve exactly in $3$ dimensions, and $\mathcal{NP}$-hard to approximate within $n^{1/2-\epsilon}$ for any $\epsilon > 0$ in $4$ or more dimensions. 
We also give a new algorithm for $2$ dimensions, albeit with the same $O(n^3)$ running time complexity as that of the algorithm of Bautsita et al. 
\end{abstract}

\section{Introduction}

Suppose you are given a set of $n$ points $S$ in $\mathbb{R}^d$ with weights $w : S \rightarrow \mathbb{R}$; note that weights can be positive or negative. The weight of a polytope $P$ is defined as $w(P) = \sum_{x \in S \cap P} w(x)$. In the \emph{maximum weight convex polytope} problem, or $MWCP$ for short,  the goal is to find a convex polytope of maximum weight. This is a rather natural and fundamental computational geometry question. 

$MWCP$ with a binary weight function, such as $w : S \rightarrow \{+1, -1\}$, belongs to a large class of computational geometry problems on bichromatic point sets with weights $\{+1, -1\}$ corresponding to two colors, typically ``red'' and ``blue''. For example, in the maximum box problem one is given a set of $r$ red points and a set of $b$ blue points in the plane and the goal is to find an axis-aligned rectangle which maximizes the number of blue points and does not contain any red points. Liu and Nediak \cite{liu2003planar} gave an exact $O(r\log r+r+b^2\log b)$ algorithm, and Eckstein et al. \cite{eckstein2002maximum} construct an efficient branch-and-bound algorithm motivated by a problem in data analysis. Liu and Nediak \cite{liu2003planar} also show how to solve efficiently a related bichromatic separability with two boxes problem, introduced by Cort{\'e}s et al. \cite{cortes2009bichromatic}.

$MWCP$ is also related to bichromatic discrepancy problems, where one is given two finite sets of points $S^+$ and $S^-$ in $\IR^d$, and the goal is to find an axis aligned parallelepiped (also called a box) $B$ maximizing the difference between the number of the points of $S^+$ and $S^-$ inside the box, i.e. $||B\cap S^+|-|B\cap S^-||$. Let $n=|S^+\cup S^-|$ denote the total number of points. Dobkin et al. \cite{dobkin1996computing} solved this problem in $\IR^2$ in $O(n^2\log n)$ time. Liu and Nediak \cite{liu2003planar} presented a $2$-factor approximation for this problem in $\IR^2$ with $O(n \log^2 n)$ running time. 

In another related problem, namely, numerical discrepancy problem, one is given a set of $n$ points $S \subset [0,1]^2$. The goal is to find a box $B$ that maximizes the numerical discrepancy of $B$ defined as $||B \cap S|/|S| - \mu(B) |$, where $\mu(B)$ denotes the area of $B$. Observe that the numerical discrepancy of $B$ can be thought of as measuring the deviation of the empirical distribution from the uniform distribution. Dobkin et al.  \cite{dobkin1996computing} solved this problem in $\IR^2$ in $O(n^2 \log^2 n)$ time. Liu and Nediak \cite{liu2003planar} presented a 2-factor approximation for this problem in $\IR^2$ with $O(n \log^3 n)$ running time.

The above problems introduce constraints on the shape of the solution, namely that the convex polygon must be an axis-aligned parallelepiped. In another variation studied by Gonz{\'a}lez-Aguilar et al. \cite{gonzalez2019maximum} the geometric shape of the solution is restricted to be a rectilinear convex hull of points (note that the rectilinear convex hull is not necessarily a convex subset of $\mathbb{R}^2$). Gonz{\'a}lez-Aguilar et al. \cite{gonzalez2019maximum} gave an $O(n^3)$ algorithm for this problem.


We note that the above problems are very similar to our problem at first glance. A deeper investigation shows that the nature of restriction on the solution set is crucial for the above problems and algorithms for them, and so new ideas and techniques are needed for $MWCP$ problem. There is one other problem that is directly relevant to $MWCP$, and that is the optimal islands problem studied by Bautista et al. \cite{bautista2011computing}. In this problem, one is given a set $S$ of $n$ points colored with $2$ colors in the plane. A subset $\mathcal{I} \in S$ is called an island of $S$, if $\mathcal{I}$ is an intersection of $S$ and a convex set $C$. Bautista et al. \cite{bautista2011computing} gave an $O(n^3)$-time algorithm to find a monochromatic island of maximum cardinality. Their algorithm can also be used to solve the $MWCP$ problem in $2$ dimensions. 

The class of problems to which $MWCP$ belongs have important practical applications in data analysis and machine learning. In particular, Bautista et al. \cite{bautista2011computing} were motivated by clustering applications. Given a training dataset of points $S \subset \mathbb{R}^d$ that are labelled with two colors ``red'' and ``blue'', in a classification problem one is interested in a simple description of a region of space corresponding to the class of ``red'' points, for example. One possibility is to use convex hulls for such a description (see, for example, Kudo et al. \cite{kudo1998approximation}). 
If dataset is $2$-dimensional one arrives naturally at the optimal islands problem. However, datasets are often noisy, so one should not expect to see large monochromatic islands, so perhaps weighted version of the problem, such as $MWCP$, might be more suitable. A bigger issue is that in classification problems datasets are often high dimensional and one cannot always hope to obtain clusters by projecting to $2$ dimensions first. 
Thus, for clustering applications it is important to be able to solve $MWCP$ efficiently in high dimensions. This is the question we tackle in this paper. Alas, we show that $MWCP$ is $\mathcal{NP}$-hard in $3$ dimensions (Theorem~\ref{theorem: 3d hardness}), and that it is $\mathcal{NP}$-hard to approximate within $n^{1/2-\epsilon}$ for any $\epsilon > 0$ in $4$ dimensions even with binary weights (Theorem~\ref{theorem: 4d hardness}). We also give a completely new algorithm for $2$ dimensions with running time $O(n^3)$ matching Bautista et al.

\section{Preliminaries}

Whenever we write ``polytope'' in this paper we mean a convex polytope. $S$ denotes the input set of $n$ points in $\IR^d$ for $d\geq 1$ and a weight function is denoted by $w:S\to \IR$. The weight of a polytope $P$, denoted by $w(P)$, is defined as follows:
$$w(P)= \sum_{v\in S \cap P} w(x)$$
In $MWCP$ problem, the goal is to find a polytope with maximum weight. Note that points $v \in S$ with $0$-weight do not affect weight of any polytope, and so they can be removed from the input in a preprocessing step. Henceforth, we assume that for all $v \in S$ we have $w(v) \neq 0$. We use $S^-$ and $S^+$ for the subsets of points of $S$ with negative and positive weights respectively. For a set of points $C \subset \mathbb{R}^d$ we let $conv(C)$ denote the convex hull of $C$. With a slight abuse of notation, we define $w(C)=w(conv(C))$. A subset $C\subseteq S^+$ is \textit{maximal} if for every $v\in C$, $w(C)>w(C\setminus \{v\})$. 

Recall that a polytope has two standard equivalent descriptions: $\mathcal{V}$-polytope is described as a convex hull of vertices, and $\mathcal{H}$-polytope is described as an intersection of half-spaces. We shall primarily work with $\mathcal{V}$-polytopes due to the nature of $MWCP$ problem. We let $vert(P)$ denote the set of vertices of a polytope $P$.  Vertices of a polytope are also its $0$-faces and edges of a polytope are its $1$-faces. We state a few facts about polytopes here that will be used later in the paper; for a more thorough introduction to polytope theory, the reader is referred to the excellent lecture notes of Ziegler~\cite{ziegler2012lectures}.

\begin{fact} [$\mathcal{V}$-polytope definition]
\label{fact: convex combination}
Let $P\subseteq \IR^d$ be a polytope and $v\in \IR^d$ be a point. $v\in P$ if and only if there is a convex combination of $vert(P)$ equal to $v$.
\end{fact}

\begin{fact}
\label{fact: face is polytope}
Let $P\subseteq \IR^d$ be a polytope and $F$ be a face of $P$. The face $F$ is a polytope, with $vert(F) = F \cap vert(P).$
\end{fact}
Let $P\subseteq \IR^d$ be a polytope and $F$ be a face of $P$. For a hyperplane $h$ such that $F\subseteq h$ we define $h^-$ and $h^+$ to be the open half spaces bounded by $h$ such that $h^-\cap P = \emptyset$.




A polytope $P\in \IR^d$ is a \textit{polytope embedding} of a graph $G(V,E)$ if there exist a one-to-one function $f:V\to vert(P)$ such that if $(u,v) \in E$ then $(f(v),f(u))$ is an edge of $P$. Note that $P$ may have some extra edges compared to $G$. If $P$ has exactly $|E|$ edges, then we call this embedding a \textit{polytope realization} of $G$.

\section{Results}

In this section we present our results for the $MWCP$ problem beginning with an overview of upper bounds in Section~\ref{sec:ubs} (where we present a new algorithm for $2$ dimensions), followed by lower bounds for $3$ and $4$ dimensions in Section~\ref{sec:lbs}.

\subsection{Upper bounds for \texorpdfstring{$1$}{Lg} and \texorpdfstring{$2$}{Lg} dimensions}
\label{sec:ubs}

We begin with a simple observation: we can assume without loss of generality that vertices of a maximum weight polytope are elements of $S^+$.

\begin{lemma}
\label{lemma: positive enough}
For every set $S$ of points  in $\IR^d$, there exists a maximum weight polytope $P$ with $vert(P) \subseteq S^+$.
\end{lemma}

\begin{proof}
Let $P$ be a maximum weight polytope and define $C= S^+\cap P$. The convex hull $conv(C)$ is a subset of $P$ that has all the positive points of $P$. Thus, $w(C)\geq w(P)$. Since $vert(conv(C)) \subseteq S^+$, we have that $conv(C)$ satisfies the conditions of the lemma.
\end{proof} 

The above lemma implies that to solve $MWCP$ it is sufficient to find a set $C\subseteq S^+$ with maximum weight of its convex hall. In particular, when $d = 1$ the $MWCP$ problem reduces to the maximum subarray problem (consider the array of weights of points in $S$ in increasing order of their $x$-coordinates). The following result is immediate from well known algorithms for the maximum subarray problems.

\begin{theorem}
The $MWCP$ problem in $1$ dimension ($d = 1$) is solvable in $O(n \log n)$ time. Moreover, if input points are sorted the problem is solvable in $O(n)$ time.
\end{theorem}


Bautista et al.~\cite{bautista2011computing} gave a dynamic programming algorithm that solves the $MWCP$ problem in $2$ dimensions in $O(n^3)$ time. Their algorithm is based on a triangulation of a convex polytope from a topmost \emph{anchor} vertex. 

\begin{theorem}[Bautista et al. \cite{bautista2011computing}]
The $MWCP$ problem is solvable in $O(n^3)$ time in $2$ dimensions ($d =2$).
\end{theorem}

In the rest of this section we present a new algorithm which solves $MWCP$ problem in $2$ dimensions, albeit with the same $O(n^3)$ running time. Our algorithm is based on a different decomposition (see Figure~\ref{fig:two-decompositions}), and is arguably simpler than the algorithm of Bautista et al.

\begin{figure}[H]

\centering
\includegraphics[scale=0.7]{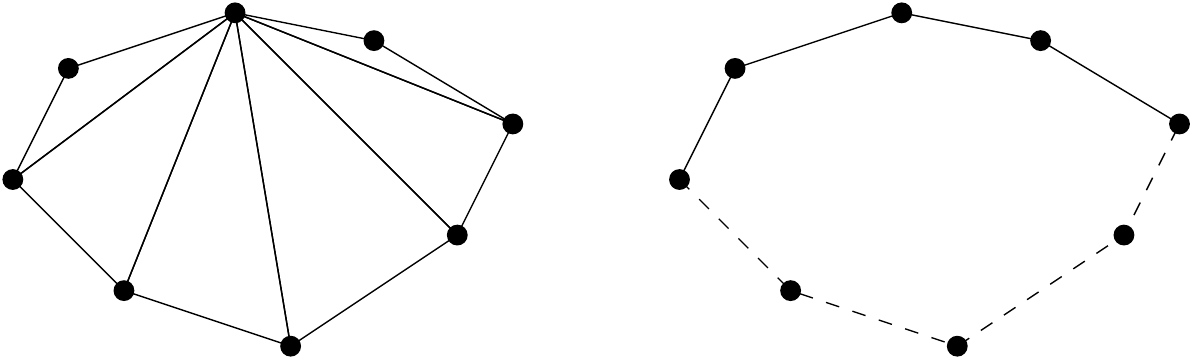}
\caption{Two decompositions of a polytope which form a basis of two dynamic programming approaches. In the approach of Bautista et al. (shown on the left) a polytope is decomposed via a triangulation from an anchor (topmost) vertex. In our approach (shown on the right) a polytope is decomposed into two paths from a leftmost to a rightmost vertex: top concave path (shown solid) and bottom convex path (shown dashed).}\label{fig:two-decompositions}
\end{figure}


Without loss of generality we can assume that no two points of $S$ have the same $x$-coordinates. Otherwise in $O(n^2)$ we can find line $\ell$ such that is not parallel to any line passing through two point in $S$. Then we can rotate the axes so that the $y$-axis becomes parallel to $\ell$.

Let $p_1,\ldots,p_n$ be the points in $S$ sorted from left to right by their $x$-coordinates. Consider a directed edge from $p_i$ to $p_j$ for every $i<j$. Weight of the edge $p_i\to p_j$, denoted by $w(p_i,p_j)$, is the sum of all the weights of points $p_k$ such that $i < k < j$ and $p_k$ is below the line segment joining $p_i$  and $p_j$. We can use brute-force algorithm to compute $w(p_i,p_j)$ for all $i<j$ in $O(n^3)$ time. Thus, we assume that all these weights have been precomputed and are available to us when we need them. A \textit{path} is a sequence of connected edges. For a path $\mathcal{P}$ we define its weight, denoted by $w(\mathcal{P})$, to be the sum of the weights of its edges and its vertices. For a path $\mathcal{P}$ we define its sub-weight, denoted by $w^{-}(\mathcal{P})$, to be the sum of the weights of its edges only.

A polygon $P$ can be represented as a concave path $\mathcal{C}$ and a convex path $\mathcal{V}$ between its leftmost and its rightmost vertices (see Figure~\ref{fig:two-decompositions}). Thus the weight of $P$ is equal $w(\mathcal{C})-w^{-}(\mathcal{V})$. We shall present a dynamic programming algorithm to solve the optimization version of the problem, where we are interested in computing the weight of a maximum weight polygon only. The algorithm can be easily modified to find a maximum weight polygon itself by the standard technique of remembering which choices resulted in individual entries of the dynamic programming tables. 

For every $i<j\leq k$, let $C[i,j,k]$ (respectively $V[i,j,k]$) be the maximum (respectively, minimum) weight (respectively, sub-weight) of a concave (respectively, convex) path from $p_i$ to $p_k$ such that the first edge is $p_i \to p_j$. We denote the maximum weight of a polygon with leftmost vertex $p_i$ and rightmost vertex $p_k$ by $M[i,k]$. If $i = k$ then $M[i,k] = w(p_k)$, and if $i < k$ then $M[i,k]$ can be computed as:
\[M[i,k] = \max_{j:i<j\leq k} C[i,j,k] - \min_{j:i<j\leq k} V[i,j,k].\]
The solution to the overall problem is then given by the $\max_{i \le k} M[i,k]$.

In the remainder, we explain how the table $C[i,j,k]$ can be computed. The table $V[i,j,k]$ is computed analogously with some trivial modifications (such as excluding contribution of vertices of the path, replacing concavity with convexity, and replacing maximization objective with minimization objective).

In the algorithm, we have to check whether a line segment joining vertices $p$ and $q$ can be extended to a vertex $r$ with $p.x < q.x < r.x$ while maintaining concavity. This can be tested by checking whether the vector $r-p$ is turned clockwise relative to the vector $q-p$ (see Figure~\ref{fig:concavity-test}). In turn, this can be achieved by checking the sign of $2$-dimensional cross-product, denoted by $\times_2$, and defined as $v_1 \times_2 v_2 = v_1.x \cdot v_2.y - v_1.y \cdot v_2.x$. To summarize we have that the path $p \to q \to r$ is concave if and only if\footnote{A bit of care is needed to handle inputs that are not in general position. If three points $p, q, r$ with $p.x < q.x < r.x$ are collinear then $(r-p) \times_2 (q-p) = 0$, and the path $p, q, r$ should be considered concave. However, this makes $q$ not a vertex of the resulting polytope, as it appears in the middle of an edge. In our description, we tacitly assumed that points are in general position to simplify the presentation. It is easy to extend our algorithm to handle points not in general position.}$(r-p) \times_2 (q-p) > 0$.

\begin{figure}[H]

\centering
\includegraphics[scale=0.7]{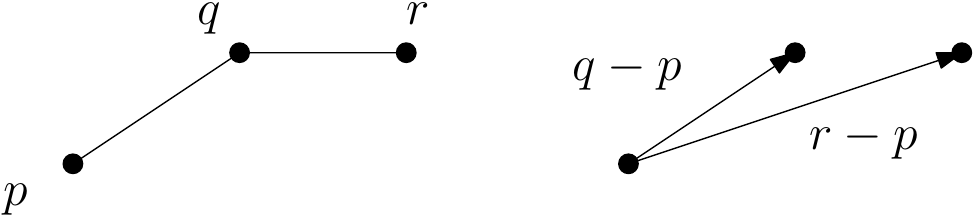}
\caption{The path $p \to q \to r$ is concave if and only if vector $r-p$ is turned clockwise relative to vector $q-p$.}\label{fig:concavity-test}
\end{figure}

Base cases for the table $C[i,j,k]$ are the following:
\begin{align*}
    C[i, k, k] &= w(p_i, p_k) + w(p_i) + w(p_k) \qquad \text{if } i < k \\
    C[i, j, k] &= -\infty \qquad \text { if } i < j < k   \\
    &\qquad\qquad\text{ and } (p_k - p_i) \times_2 (p_j - p_i) < 0
\end{align*}

It is clear that the other entries $C[i,j,k]$ with $i < j < k$ can be computed according to the following formula:
\begin{align}
    C[i,j,k] &= \max_{j'} \{ w(p_i, p_j) + w(p_i) + C[j, j', k] : \label{eq:main_computation}\\
    &  j < j' \le k \text { and } (p_{j'}-p_i) \times_2 (p_j-p_i) > 0\}.\nonumber
\end{align} 

A naive computation of the above table takes $O(n^4)$ time, since the table has $O(n^3)$ entries and each entry can be computed in $O(n)$ time. Next, we show a trick of how the time complexity can be reduced to $O(n^3)$. The idea is for a fixed $j$ and $k$ to fill in entries $C[i,j,k]$ for all $i$ in $O(n)$ time.

We precompute in $O(n^2 \log n)$ total time for all $j$ two lists: $L_j= (l_1,\ldots,l_{j-1})$ and $R_j=(r_1,\ldots,r_{n-j})$. $L_j$ ($R_j$) consists of points $\{p_1,\ldots, p_{j-1}\}$ (respectively, $\{p_{j+1},\ldots, p_n\}$) to the left (respectively, to the right) of $p_j$ and sorted in clockwise order with respect to $p_j$ as the origin. 

Now, fix a pair of indices $j < k$. In $O(n)$ time it is easy to compute $D[j',k] = \max_{j''} \{ C[j, j'', k] : j'' \le k \text{ and } p_{j''} \text{ is either $p_{j'}$ or appears after } p_{j'} \text{ in } R_j\}$. Define the first compatible $j'$ for the given $i,j$, denoted by $fc(i,j)$, as the first $p_{j'}$ appearing in $R_j$ such that $p_i \to p_j \to p_{j'}$ is concave. Then it is clear that $C[i,j,k]$ can be equivalently restated as follows:
\[ C[i,j,k] = w(p_i, p_j) + w(p_i) + D[fc(i,j), k].\]
This is because, every $p_{j''}$ that appears after $fc(i,j)$ in $R_j$ also forms a concave path $p_i \to p_j \to p_{j''}$. Thus, the third term $D[fc(i,j), k]$ in the above equation is exactly the same as the third term in Equation~\eqref{eq:main_computation}.

Lastly, it is left to observe that as one considers points $p_i$ in the order in which they appear in $L_j$, the corresponding sequence of $fc(i,j)$ also forms an increasing sequence in $R_j$. Thus, by maintaining a running pointer into $R_j$ one can compute $fc(i,j)$ in $O(n)$ time for all $p_i \in L_j$. This finishes the description of the algorithm. One readily checks that all precomputing steps take $O(n^3)$, base cases of $C[i,j,k]$ can also be computed in $O(n^3)$ time, and all other entries can be computed in $O(n^3)$ as well, by iterating over all pairs $j < k$ and filling in $C[i, j,k]$ for all $i$ in $O(n)$ time.

\subsection{Lower bounds for \texorpdfstring{$3$}{Lg} and \texorpdfstring{$4$}{Lg} dimensions}
\label{sec:lbs}

Recall that a \emph{strict reduction} from an optimization problem $\mathcal{A}$ to an optimization problem $\mathcal{B}$ is a pair of functions $(f, g)$, where $f$ maps instances $x$ of $\mathcal{A}$ to instances $f(x)$ of $\mathcal{B}$ and $g$ maps solutions $y$ of $\mathcal{B}$ to solutions $g(y)$ of $\mathcal{A}$, such that the approximation ratio achieved by solution $y$ on instance $f(x)$ of $\mathcal{B}$ is at least as good as the approximation ratio achieved by solution $g(y)$ on instance $x$ of $\mathcal{A}$. 
All our lower bound results in this section are based on the following technical lemma.


\begin{lemma}
\label{lemma: general reduction}
Let $\mathcal{G}$ be a graph family. If for every $G \in \mathcal{G}$ a polytope embedding of $G$ into $\mathbb{R}^d$ can be found in polynomial time and bit complexity polynomial in $n$, then there is a strict reduction from the maximum independent set on $\mathcal{G}$ to $MWCP$ in $d$ dimensions with weights $\{+1, -1\}$.
\end{lemma}


\begin{proof}
Given input instance $G = (V,E)$ to the maximum independent set on $\mathcal{G}$, we let $P$ be the result of applying the polytope embedding to $G$. 
Let $S^+:=vert(P)$ and assign $+1$ weight to every vertex in $S^+$. Create set $S^-$ by adding two points with weights of $-1$ at two arbitrary positions of every graph edge. Let $S=S^+\cup S^-$. For a negative point $v\in S^-$, let $p_1(v),p_2(v)\in S^+$ be positive-weighted vertices such that $v$ was placed on the edge joining $p_1(v)$ with $p_2(v)$ and $n(v)$ be the other negative point on that edge. See Figure~\ref{fig:embedding} for an example.

We claim that for a subset $C\subseteq S^+$, there exist a negative point $v\in S^-$ in $conv(C)$ if and only if $p_1(v), p_2(v) \in C$. One direction is clear: if $p_1(v), p_2(v) \in C$ then by Fact~\ref{fact: convex combination} $n(v)$ and $v$ are in $conv(C)$.  For the other direction, assume that $v \in conv(C)$. Let $e$ be the edge between $p_1(v)$ and $p_2(v)$. By the definition of $P$, there exist a hyperplane $h_{e}$ such that $S^+\cap h^-_e = \emptyset$.  Therefore $C\cap h^-_e = \emptyset$ and $F:=conv(C)\cap h_e$ is a face of $conv(C)$. $F\neq \emptyset$ since $v$ is in $h_e$ and $conv(C)$. Only vertices of $S^+$ in $h_e$ are $\{p_1(v),p_2(v)\}$. By Fact \ref{fact: face is polytope} $vert(F)= F\cap vert(conv(C)) \subseteq h_e \cap S^+ = \{p_1(v),p_2(v)\}$. 
Without loss of generality suppose $vert(F)=\{p_1(v)\}$, this implies $v\notin F$ which is a contradiction. Thus $vert(F)=\{p_1(v),p_2(v)\}$ and $p_1(v),p_2(v) \in C$. 

Let $C\subseteq S^+$ be a maximal subset. We claim 
that $conv(C)$ contains no negative points and all positive points in $conv(C)$ are precisely the vertices of $conv(C)$.
First, suppose there exists a negative point $v\in conv(C)$ thus $p_1(v),p_2(v)\in C$ and $n(v)\in conv(C)$. $w(C\backslash \{p_1(v)\}) \geq w(C)+2-1 > w(C)$ since $v,n(v),p_1(v)\notin conv(C\backslash \{v_i\})$. This is a contradiction to maximality of $C$.
Second, suppose there exist a positive point $v\in S$ in $conv(C) \setminus vert(conv(C))$. Because $v$ is a vertex of $P$ there exist a hyperplane $h_v$ such that $S^+\cap h^-_v = \emptyset$. Therefore $C\cap h^-_e = \emptyset$ and $v$ is a vertex of $conv(C)$ which is a contradiction.

Therefore, we can conclude that $w(C) = |C|$ if $C\subseteq S^+$ is maximal. Next, we prove there exists a maximal subset $C \subseteq S^+$ if and only if there exist an independent set $\mathcal{I} \subseteq V$ such that $w(C)=|\mathcal{I}|$.

\textbf{If:} Let $\mathcal{I} \subseteq V$ be an independent set and $C\subseteq S^+$ be the set of corresponding vertices of $\mathcal{I}$ in $S^+$. Because there is no edge between vertices in $\mathcal{I}$, there is no graph edge between vertices in $C$. Thus there are no negative points in $conv(C)$. Since all vertices inside $conv(C)$ are positive, $C$ is a maximal subset and $w(C)=|C|=|\mathcal{I}|$.

\textbf{Only if:} Let $C\subseteq S^+$ be a maximal subset and let $\mathcal{I} \subseteq V$ be the set of corresponding vertices of $C$ in $G$. Because $C$ is a maximal subset, there is no negative point in $conv(C)$, and there is no graph edge between vertices of $C$. Thus the set of corresponding vertices of $C$ in $G$ is an independent set. $|\mathcal{I}|=w(C)$ since $w(C)=|C|$.

Without loss of generality we can suppose every approximation algorithm for $MWCP$ outputs a maximal subset of $S^+$. Thus there exist a strict reduction from the maximum independent set problem of graph $G(V,E)$ to $MWCP$ in $\IR^d$.
\end{proof}

\begin{figure}[H]

\centering
\includegraphics[scale=0.06]{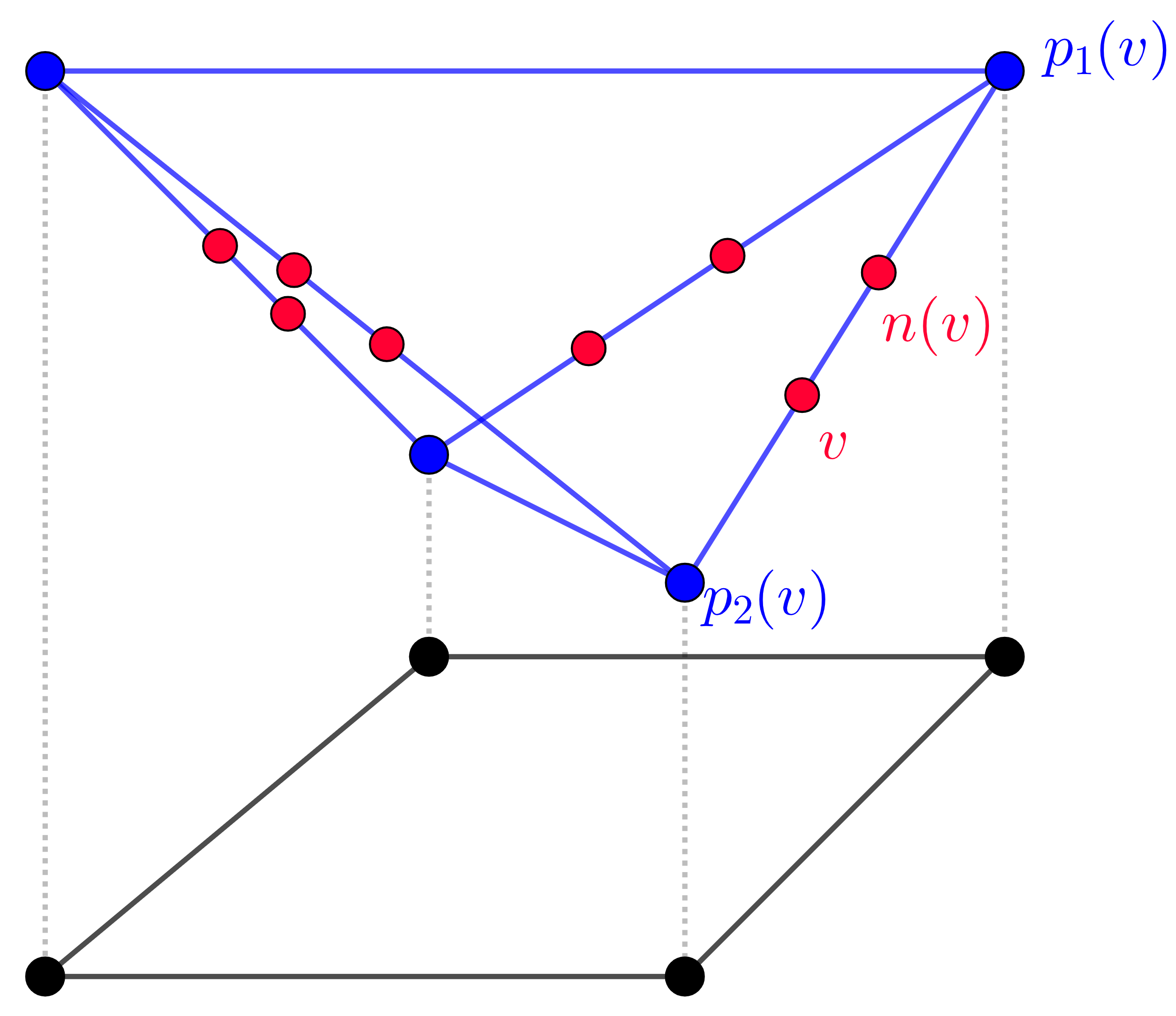}
\caption{A graph and its embedding in $\IR^3$. Black points and edges are the graph and blue points and edges are the embedding of the graph. Red points are added negative points. And an example of $v$, $n(v)$, $p_1(v)$, and $p_2(v)$ is shown.}\label{fig:embedding}
\end{figure}


We obtain the lower bound for $3$ dimensions by applying Lemma~\ref{lemma: general reduction} to the class $\mathcal{G}$ of planar graphs. We note that the maximum independent set problem is $\mathcal{NP}$-hard even for planar graphs \cite{garey1979computers}. 
Our lower bound relies on the polynomial embedding in $3$ dimensions due to Das et al.~\cite{das1997complexity}. A \textit{maximal planar graph} is a planar graph such that an addition of any new edge results in a non-planar graph. 

\begin{lemma}[Das et al. \cite{das1997complexity}]
\label{lemma: planar to polytope}
Given a maximal planar graph $G(V,E)$ with $n$ vertices, a polytope realization of $G$ in $\IR^3$ can be found in $O(n)$ time and with bit complexity polynomial in $n$.
\end{lemma}

Thus the following theorem can be easily deduced from Lemmas~\ref{lemma: planar to polytope} and~\ref{lemma: general reduction}. 


\begin{theorem}
\label{theorem: 3d hardness}
Let $S$ be a set of $n$ points in $\IR^3$ with weight function $w$, finding $MWCP$ of $S$ is $\mathcal{NP}$-hard even if $w: S\to \{-1,+1\}$.
\end{theorem}

\begin{proof}
Let $\mathcal{G}$ be the family of all planar graphs. By adding edges to a planar graph $G$ we can make it maximal. The polytope realization of the new maximal planar graph is also a polytope embedding of $G$. Thus with Lemma~\ref{lemma: planar to polytope} we can conclude for every $G \in \mathcal{G}$ a polytope embedding of $G$ in $\IR^3$ can be found in polynomial time and with polynomial bit complexity. By Lemma~\ref{lemma: general reduction}, there is a strict reduction from maximum independent set on planar graphs to $MWCP$ with weights $\{+1, -1\}$, hence it is an $\mathcal{NP}$-hard problem.
\end{proof}

Let $S$ be the set of points $(i,i^2,i^3,i^4)$ for $1\leq i \leq n$ in $\IR^4$. The convex hull of $S$ is known as the cyclic polytope on $n$ vertices in $\IR^4$ and it is a polytope realization of a complete graph with $n$ vertices (for more details, see, for example, \cite{ziegler2012lectures}).

\begin{lemma}
\label{lemma: complete to polytope}
Given a complete graph $K_n$ with $n$ vertices, a polytope realization of it in $\IR^4$ can be found in $O(n)$ time with a bit complexity polynomial in $n$.
\end{lemma}

We can use Lemma \ref{lemma: complete to polytope} to show that $MWCP$ in $4$ dimensions is as hard as independent set on arbitrary graphs. Zuckerman~\cite{zuckerman2007}, strengthening an earlier result of H\r{a}stad~\cite{hastad96}, showed that it is $\mathcal{NP}$-hard to approximate independent set on arbitrary graphs within $n^{1-\epsilon}$ factor for any $\epsilon > 0$.

\begin{theorem}[Zuckerman~\cite{zuckerman2007}]
\label{thm:indep-set-hardness}
For any $\epsilon > 0$ it is $\mathcal{NP}$-hard to approximate maximum independent set to within $n^{1-\epsilon}$.
\end{theorem}

 Combining the above ingredients we establish the inapproximability of $MWCP$ in $4$ dimensions and higher.
 
\begin{theorem}
\label{theorem: 4d hardness}
For any $\epsilon > 0$ it is $\mathcal{NP}$-hard to approximate $MWCP$ in $4$ dimensions (or higher) with weights $\{+1, -1\}$ to within $n^{1/2-\epsilon}$.
\end{theorem}

\begin{proof}
Let $\mathcal{G}$ be the family of all finite graphs. By Lemma~\ref{lemma: complete to polytope} for every $G \in \mathcal{G}$ a polytope embedding of $G$ in polynomial time and with polynomial bit complexity can be found (recall that the embedding is allowed to have extra edges compared to $G$). By Lemma~\ref{lemma: general reduction}, there is a strict reduction from maximum independent set on general graphs to $MWCP$ with weights $\{+1, -1\}$. Since Theorem~\ref{thm:indep-set-hardness} is expressed in terms of input size, it is left to observe that the reduction of Lemma~\ref{lemma: general reduction} produces instances of $MWCP$ with the number of points that is at most quadratic in the number of vertices of the input graph.
\end{proof}

\section{Conclusion and Discussion}   

In this work, we extended our understanding of the complexity of $MWCP$ as a function of the ambient dimension $d$. Based on our work and previous work of Bautista et al.~\cite{bautista2011computing}, the following picture emerges: 

\begin{enumerate}
    \item For $d = 1$, $MWCP$ is solvable in $O(n \log n)$ time exactly (simple observation);
    \item For $d = 2$, $MWCP$ is solvable in $O(n^3)$ time (Bautista et al.~\cite{bautista2011computing} with another algorithm presented in this work);
    \item  For $d = 3$, $MWCP$ is not solvable in polynomial time unless $\mathcal{P} = \mathcal{NP}$ (this work);
    \item For $d \ge 4$, $MWCP$ is $\mathcal{NP}$-hard to approximate to within $n^{1/2-\epsilon}$ for any $\epsilon > 0$ (this work).
\end{enumerate} 
The above list immediately suggests several open problems, the following two of which are of particular interest:
\begin{openprob}
Find an algorithm with better time complexity than $O(n^3)$ for $MWCP$ in $2$ dimensions or prove a lower bound probably with some fine-grained hypothesis.
\end{openprob}
\begin{openprob}
Determine if $MWCP$ can be approximated within a constant factor in $3$ dimensions.
\end{openprob}
We conjecture that the answer to the first open problem is that there is no algorithm significantly faster than $O(n^3)$. 
In light of the second open problem, it is tempting to consider what approximation guarantees are provided by polytopes with constantly many vertices. As the following result demonstrates, constant approximation cannot be guaranteed by such solutions even in $2D$.

\begin{theorem}
\label{thm:no-fixed-k}
By restricting solutions to polytopes with constant number of vertices one can not achieve a constant factor approximation for $MWCP$ even in $\IR^2$ and even for $\{+1, -1\}$ weights.
\end{theorem}

\begin{proof}
Let $P$ be a regular $n$-gon and let the weight of each vertex be $+1$. Put a vertex with weight $-1$ outside of $P$ on the perpendicular bisector of each edge of $P$ at $\epsilon$-distance away from the edge. Choose $\epsilon$ so that line segments joining every two consecutive negative points cross $P$. This defines the instance of $MWCP$ with $P$ being an optimal solution of weight $n$. 

 Let $v_1,v_2,\dots,v_n$ and $u_1,u_2,\dots u_n$ be vertices of the clockwise order of $S^+$ and $S^-$, respectively, such that $u_i$ has $\epsilon$-distance with the edge between $v_i$ and $v_{i+1}$ ($v_{n+1}:=v_1$).

 Let $C$ be a convex $k$-gon, we claim $w(C)\leq k$. Observe that what makes this claim non-trivial is that we cannot assume that $vert(C) \subseteq S^+$ as in Lemma~\ref{lemma: positive enough}, since we have an additional restriction of exactly $k$ vertices.
 
 $\overline{ C\backslash P}$ (the closure of $C\backslash P$) is a set of vertices, edges and non-convex polygons. Let $C'$ be one of these non-convex polygons. 
 It suffices to show $w(C')\leq |vert(C)\cap vert(C')|$. Without loss of generality suppose $vert(C')\cap S^+ = \{v_1,v_2, ..., v_r\}$. 

 Let \textit{outer negative points} be the set $\{u_{i_1}, u_{i_2}, \dots, u_{i_\ell}\} \subseteq \{u_1, u_2, \dots u_{r-1}\}$ such that for every $1\leq j \leq \ell$, $u_{i_j} \notin C'$. For each $1\leq j \leq \ell$ associate $u_{i_j}$ to the edge $e$ of $C'$ that crosses the shortest line between $u_{i_j}$ and $P$. By the choice of $\epsilon$ two vertices of $e$ are in $vert(C') \cap vert(C)$ and no edge is associated to more than one outer negative point. Thus $|vert(C') \cap vert(C)| \geq \ell +1$. On the other hand there is at most $r$ positive and at least $r-1-l$ negative points in $C'$ thus $w(C') \leq l+1 \leq |vert(C') \cap vert(C)|$.
\end{proof}

\begin{figure}[H]

\centering
\includegraphics[scale=0.7]{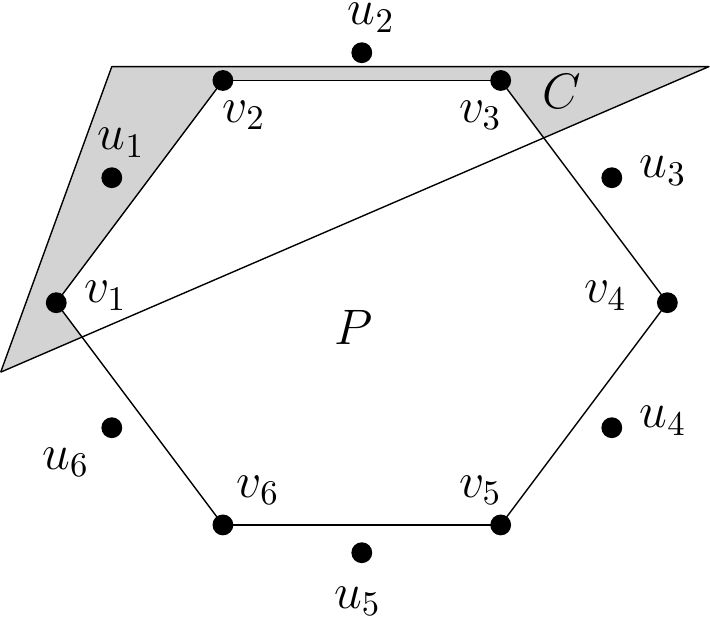}
\caption{Illustration of the proof of Theorem~\ref{thm:no-fixed-k}. Here, $n = 6, k = 3$, we chose $C$ to result only in a single $C'$, which is shown as a shaded area. We have $l = 1$ with $u_{i_1} = u_2$ and vertex $u_2$ is associated with the topmost edge of $C'$. We have $w(C') = w(v_1) + w(v_2) + w(v_3) + w(u_1) = 3 - 1 = 2 = l+1$.}\label{fig:no-fixed-k}
\end{figure}



\small
\bibliographystyle{abbrv}



\bibliography{bibs/full,bibs/refs}
 
\newpage


\end{document}